\documentclass{amsart}[12pt]
\usepackage{amsmath}
\usepackage{graphicx}
\usepackage[titletoc,toc]{appendix}
\usepackage{latexsym}
\usepackage{hyperref}
\usepackage{graphicx}
\usepackage{amsthm,amssymb,amstext,amscd,amsfonts,amsbsy,amsxtra,latexsym}
\usepackage{wrapfig}
\usepackage{color}
\usepackage{pdflscape}
\usepackage{lscape}
\usepackage{tabularx}
\usepackage{lipsum}
\usepackage{rotating} 
\usepackage{pifont}
\usepackage{upgreek}
\usepackage{mathrsfs}
\usepackage{wasysym}
\usepackage[top=2.5cm, bottom=2.5cm, left=2.5cm, right=2.5cm]{geometry}

\newcommand{\lb}{\label}
\newcommand{\bea}{\begin{eqnarray}}
\newcommand{\eea}{\end{eqnarray}}

\newcommand{\lp}[2]{\Vert \, #1 \, \Vert_{#2}}

\makeatletter

\renewcommand{\tocsubsection}[3]{%
  \indentlabel{\@ifnotempty{#2}{\hspace*{2.3em}\makebox[2.3em][l]{%
    \ignorespaces#1 #2.\hfill}}}#3}
\renewcommand{\tocsubsubsection}[3]{%
  \indentlabel{\@ifnotempty{#2}{\hspace*{4.6em}\makebox[3em][l]{%
    \ignorespaces#1 #2.\hfill}}}#3}
		
\newcounter{mnotecount}[section]

\renewcommand{\themnotecount}{\thesection.\arabic{mnotecount}}

\newcommand{\mnote}[1]
{\protect{\stepcounter{mnotecount}}$^{\mbox{\footnotesize $%
\!\!\!\!\!\!\,\bullet$\themnotecount}}$ \marginpar{
\raggedright\tiny\em $\!\!\!\!\!\!\,\bullet$\themnotecount: #1} }

\newcommand{\metric}{\mathbf{g}}
\newcommand{\nN}{\mathcal{N}}
\newcommand{\bbR}{\mathbb{R}}
\newcommand{\Lx}{L^2_{\sigma}}

\theoremstyle{plain}
\newtheorem{theorem}{Theorem}[section]

\newtheorem{lemma}[theorem]{Lemma}

\theoremstyle{definition}

\newtheorem{remark}[theorem]{Remark}

\numberwithin{equation}{section}

\allowdisplaybreaks[3]

\swapnumbers
\begin{document}

\begin{center}


\title[]{Semilinear wave equations on accelerated expanding FLRW spacetimes}

\author{}

\address{}

\email{}

\date{}

\parskip = 0 pt

\maketitle

\bigskip
João L.~ Costa \footnote{e-mail address: jlca@iscte-iul.pt}{${}^{,\sharp,\star}$},
Anne T.~ Franzen \footnote{e-mail address: anne.franzen@tecnico.ulisboa.pt}{${}^{,\star}$}, and 
Jes\'{u}s Oliver \footnote{e-mail address: jesus.oliver@csueastbay.edu}{${}^{,\dagger}$}

\bigskip
{\it {${}^\sharp$}Departamento de Matemática,}\\
{\it Instituto Universit\'ario de Lisboa (ISCTE-IUL),}\\
{\it Av. das For\c{c}as Armadas, 1649-026 Lisboa, Portugal}

\bigskip
{\it {${}^\star$}Center for Mathematical Analysis, Geometry and Dynamical Systems,}\\
{\it Instituto Superior T\'ecnico, Universidade de Lisboa,}\\
{\it Av. Rovisco Pais, 1049-001 Lisboa, Portugal}

\bigskip
{\it {${}^\dagger$} California State University East Bay,}\\
{\it 25800 Carlos Bee Boulevard,}\\
{\it Hayward, California, USA, 94542}

\end{center}\begin{abstract} We identify a large class of systems of semilinear wave equations, on fixed accelerated expanding FLRW spacetimes, with nearly flat spatial slices, for which we prove small data future global well-posedness. The family of systems we consider is large in the sense that, among other examples, it includes general wave maps, as well as natural generalizations of some of Fritz John's ``blow up'' equations (whose future blow up disappears, in our setting, as a consequence of the spacetime expansion). 
We also establish decay upper bounds, which are sharp within the family of systems under analysis.  
\end{abstract}


\section{Introduction}

It is well known that an accelerated expansion provides a mechanism that helps explain the high homogeneity and isotropy of the observed Universe~\cite{ringstromBook}. At the level of wave equations, on fixed accelerated expanding cosmologies, such process of attenuation of perturbations provides a favorable environment to establish future global  existence results closely related to ``fast'' decay estimates of some relevant quantities.     

In this paper we realize these expectations by identifying a ``large'' class of Cauchy problems for systems of semilinear wave equations
\begin{equation}
\label{systemMainv0}
\begin{cases}
\square_{\metric}\phi^A=a^{-2+\delta_{0\alpha}+\delta_{0\beta}}\nN^{A,\alpha \beta}_{BC}(\phi) \partial_{\alpha} \phi^B \partial_{\beta} \phi^C\;, \\
\phi^{A}(t_0,x)=\phi^{A}_0(x) \;\;, \; \partial_t\phi^{A}(t_0,x)=\phi^{A}_1(x) \;,
\end{cases}
\end{equation}
 on fixed accelerated expanding FLRW spacetimes with metric of the form
\begin{equation}
\label{metricFLRW}
    {\metric} := -dt^2 + a^2(t)\sigma_{ij}(x)dx^idx^j\;,  
\end{equation}
for which we prove small data future global well-posedness. We also establish decay upper bounds, which are sharp within the family of systems under consideration. We use the adjective ``large'' since~\eqref{systemMainv0}  includes: $i)$  general wave maps, which for small data and under the assumption of~{\em uniformly bounded geometry} of the target manifold (see Remark~\ref{rmkExs}) satisfy
\begin{equation}
\label{wave}
\square_{\metric}\phi^A=-{\metric}^{\alpha \beta}\Gamma^{A}_{BC}(\phi) \partial_{\alpha} \phi^B \partial_{\beta} \phi^C\;,
\end{equation}
where $\Gamma^{A}_{BC}$ are the Christoffel symbols of the target manifold's Riemannian metric;
$ii)$ but also includes other examples of equations that do not exhibit any particular form of null structure. A noteworthy example of the later corresponds to Fritz John's equation~\cite{john_blow}
\begin{equation}
\square_{\metric}\phi = (\partial_t\phi)^2\;.
\end{equation}
 Recall that in 1+3 dimensional Minkowski spacetime, i.e., if ${\metric}=\eta$, where $\eta$ is the flat metric, the semilinear term in John's equation is responsible for finite time blow-up of solutions arising from arbitrarily small, but non-trivial, smooth and compactly supported initial data. As we will see, as a consequence of our results, this is no longer the case if ${\metric}$ corresponds to the metric of an accelerated expanding FLRW cosmology, with nearly flat spatial slices.

Returning to wave maps, it is of interest to note that, on par with Einstein's equations, they arguably correspond to the class of geometric wave equations that have triggered the biggest developments on the geometric analysis of evolution equations. Most of the work on the field~\cite[Chapter 6]{taoDispersive} as been carried out in the context of Minkowski and perturbations thereof (as base manifolds~\footnote{The reader might find the lack of reference to the target manifold's topology strange, but please note that this is simply a manifestation of the fact that we are only considering small data problems. See Remark~\ref{rmkExs} below for more information.}) and a typical motivation for the study of such maps comes from cosmology~\cite{ringstromWaveMap,narita,cotsakis}.  

The original motivation for our work was to identify classes of nonlinear wave equations, with relevant content in cosmological modeling, exhibiting future small data global existence. Wave maps were therefore a natural starting point. However, it rapidly became clear that our techniques applied to a wider class of wave equations and, therefore, focusing only on wave maps would be an artificial restriction that would obscure the mechanisms for decay and global existence, in accelerated expanding cosmologies. The end result of our research was the identification of a nonlinear structure~\eqref{nonLinForm} that takes advantage of the knowledge gained concerning the decay rate of derivatives of solutions to the linear homogeneous wave equations in FLRW~\cite{jpj} to create a favorable setup in the semilinear setting. Recall that time derivatives of the linear solutions decay with a rate dictated by the expansion factor $a(t)$, while spatial derivatives are at best bounded and, in general, do not decay at all (see the next section for more information). The structure~\eqref{nonLinForm} is then designed to make sure that any badly decaying derivative is multiplied by a ``good'' derivative and/or by inverse powers of the expanding factor $a(t)$. This is akin to the role played by the celebrated null structure in Minkowski, discovered by Klainerman~\cite{Knull}. 
Nonetheless, although similar in spirit the direct generalization of the null structure to the FLWR setting~\footnote{By which we mean the structure obtained by replacing, in the original null structure, the flat metric $\eta$ by the FLRW metric ${\metric}$.} and the nonlinear structure~\eqref{nonLinForm} identified in this paper are quite different both in form and content; some relevant similarities and distinctions have already been presented in the examples discussed above.    

\subsection{Some basic lessons from previous works}
\label{secLessons}
In this paper we will be concerned with accelerated expanding FLRW cosmologies (see Section~\ref{setup} for more details) with spacetime topology ${\mathcal M}=\{(t,x)\,|\, t\in\bbR^+\,,\, x\in \bbR^n\}$ and we will assume that expansion occurs in the direction of increasing $t$, to which we will refer as the future direction.
Two causal/geometric consequences of the accelerated expansion that are particularly relevant to our work are the following:
\begin{itemize}
\item {\em Global in time information from local in space data}: given a fixed $x_0\in\bbR^n$, let $\gamma(t)=(t, x_0)\in\bbR^+\times\bbR^n$  be an observer that ``reaches infinity'' and let ${\mathcal D}$ be its domain of dependence. Then, ${\mathcal D}\cap \{t\geq T\}$ is compact, for all $T>0$, see Figure \ref{glit}.
\item {\em Cosmic  silence}: given two such curves $\gamma_i(t)=(t, x_i)\in\bbR^+\times\bbR^n$, $i=1,2$,  if we denote by ${\mathcal D}_i$ the corresponding domains of dependence, then, for all sufficiently big $T>0$, ${\mathcal D}_1\cap{\mathcal D}_2\cap \{t\geq T\}=\emptyset$, see Figure \ref{cosil}.
\end{itemize}
{\begin{figure}[!ht]
\centering
\includegraphics[width=0.8\textwidth]{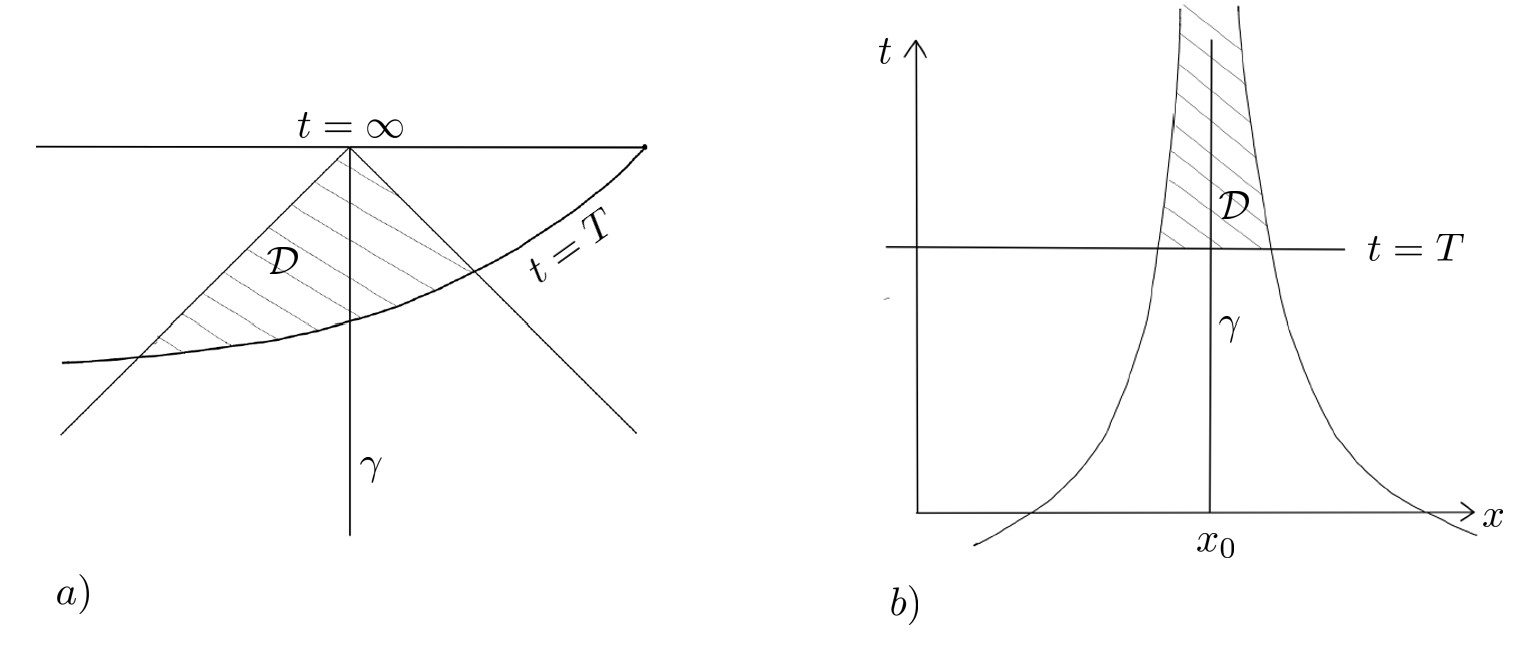}
\caption[]{a) Section of Penrose diagram with ${\mathcal D}\cap \{t\geq T\}$ depicted as the hatched region. b) 2-dimensional representation in $\bbR^+\times\bbR^n$ of ${\mathcal D}$.}
\label{glit}\end{figure}}
{\begin{figure}[!ht]
\centering
\includegraphics[width=0.8\textwidth]{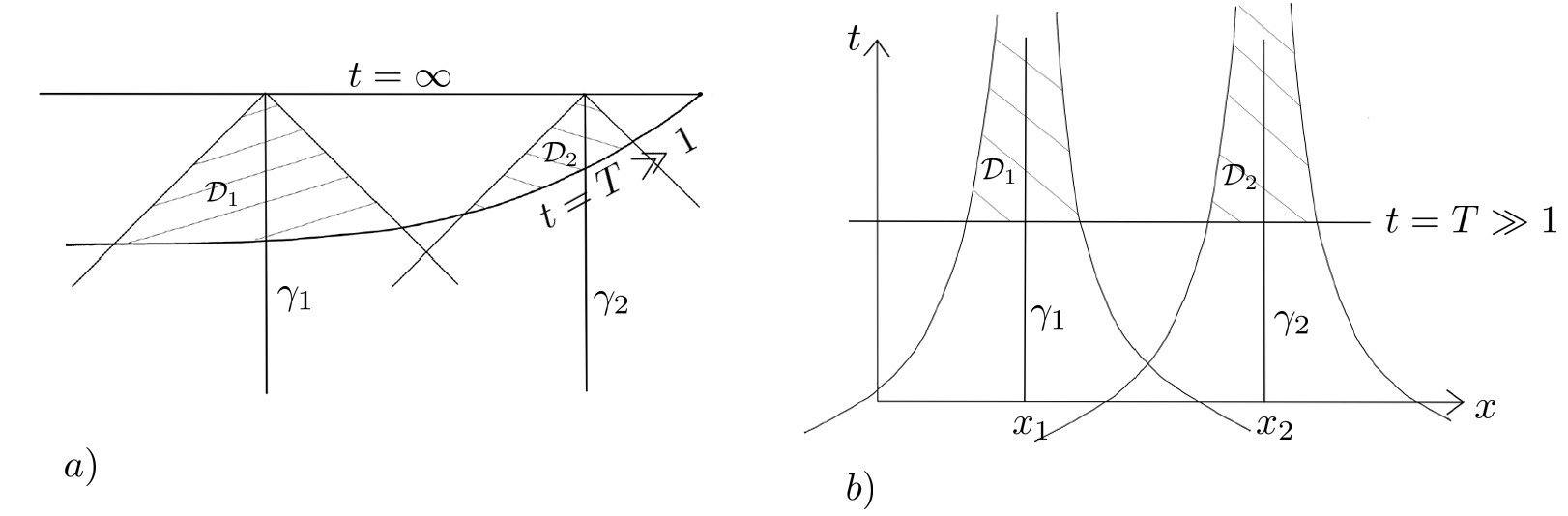}
\caption[]{a) Section of Penrose diagram with ${\mathcal D}_1\cap{\mathcal D}_2\cap \{t\geq T\}$ depicted as the hatched regions. b) 2-dimensional representation in $\bbR^+\times\bbR^n$ of ${\mathcal D_1}$ and ${\mathcal D_2}$.}
\label{cosil}\end{figure}}

As a consequence of the first property, we see that for any hyperbolic equation that, in particular, satisfies the domain of dependence property, we can obtain global in time information about its solutions, from localized initial data prescribed on a compact set of the form ${\mathcal D}\cap \{t= T\}$~\footnote{This is in stark contrast  with what happens, for instance, in Minkowski spacetime, and was used in groundbreaking work by Ringstr\"om~\cite{RingstromFuture} to establish future non-linear stability of de Sitter spacetime (a particularly relevant example of our FLRW family) as a solution of appropriate Einstein-non-linear scalar field systems.}. In particular, we can assume that our initial data is contained in a large enough torus $\mathbb T^n$. Then, if the torus is flat and we consider the homogeneous wave equation we can, in fact, derive explicit solutions using Fourier series, as done in Appendix A of~\cite{jpj}; these solutions can then be used to clarify what is the sharp asymptotic behavior of solutions.  For instance, in the case of a power law expanding factor $a(t)=t^p$, $p>1$, we will show that small data solutions in our class of semilinear wave equations satisfy the following estimates (valid in the future region)
$$|\partial_t\phi| \lesssim t^{-2p+1}$$
and 
$$|\partial_x\phi| \lesssim 1\;.$$
Since within the referred Fourier mode solutions~\cite[(155)--(157)]{jpj} there are solutions with this exact profile, with ``$\lesssim$'' replaced by ``$\sim$'', and since  the homogeneous wave equation is a particular case of our setup, we can then conclude that these estimates are sharp within our class of equations.

Now let us discuss an important consequence of {\em cosmic silence}. Let us start with the wave equation $\square_{\metric}\phi = 0$ and let us choose   $T\gg1$ such that ${\mathcal D}_1\cap{\mathcal D}_2\cap \{t\geq T\}=\emptyset$, as described before. Now consider the Cauchy problem with data posed on $t=T$ such that $\phi_{|{\mathcal D}_i\cap \{t= T\}}=C_i$ and $\partial_t\phi_{|{\mathcal D}_i\cap \{t= T\}}=0$, where the $C_i$ are distinct constants. Then, by the domain of dependence property and the fact that the homogenous wave equation admits constants as solutions, we conclude that $\phi_{|{\mathcal D}_i}(t,x)=C_i$, for all $t$. In particular, $\phi$ does not converge to a constant at infinity, instead we have $\lim_{t\rightarrow+\infty}|\phi(t,x)-\phi_{\infty}(x)|= 0$, for some (non-constant) function $x\mapsto\phi_{\infty}(x)$.  We can now easily see that the exact same conclusions apply to any (non-linear) wave equation that admits constants as solutions; this is exactly what happens with our class of semilinear wave equations (see~\eqref{decay2function1} and~\eqref{decay2function2}). Moreover this should be contrasted with what happens if we consider the Klein-Gordon case $\square_{\metric}\phi = m^2 \phi$ with non-zero mass $m$: then the only constant solution is the trivial one and the remaining solutions, arising from appropriate initial data, decay to zero at future infinity (see for instance~\cite{amol}).

\subsection{Other related works}

The mathematical analysis of wave equations on expanding cosmological spacetimes has a long and rich history that can be traced back to the work of Klainerman and Sarnak~\cite{klainerman_sarnak}. Here we will not try to give a complete overview of the subject and will instead simply focus on previous works that are concerned with the analysis of such PDEs in fixed {\em accelerated} expanding FLRW cosmologies.   

Sharp and almost sharp decay estimates for linear wave equations in accelerated expanding FLRW spacetimes, with special emphasis on de Sitter,  can be found in~\cite{ren,vasy,baskin,jpj,amol,rossetti}.
For a detailed presentation of systems of linear wave equations on various cosmological backgrounds we refer to the monograph~\cite{hans_big} of Ringström.

In~\cite{choquet} Choquet-Bruhat investigated wave maps with FLRW base space and established global existence under appropriate smallness conditions on the data and the spatial geometry.  There is some overlap between these results and the existence results of our paper in what pertains to wave maps. Nonetheless we should mention that Choquet-Bruhat's strategy is more geometric in nature and is restricted to wave maps with $1+n$ dimensional base spaces and $n\leq 3$; moreover her work does not provide an asymptotic analysis of the maps. 

A thorough study of linear and semilinear wave equations, using representation formulas via integral transforms, has been performed by Galstian and Yagdjian (see~\cite{GY,Y,Yag_surv}  and references therein). These works consider  non-linearities depending only on $\phi$  and are restricted to the Klein-Gordon case with non-vanishing mass term; this last fact reveals itself in the fact that in their case $\phi\rightarrow 0$, as $t\rightarrow 0$ (see discussion in the end of Section~\ref{secLessons}). Results along the same line have been also obtained by Ebert and Reissig~\cite{Ebert}.

\subsection{Overview}
In Section~\ref{setup} we present our  geometrical setup, the structure of our systems of semilinear equations and our main results. In Section~\ref{secEnergy} we establish our basic energy estimate using the vector field method. 
The issue of local existence is settled in Section~\ref{secLocal} where we use the conformal method to transform our equations into equations in Minkowski, where we can invoke classical local existence 
 results. The proof of small data global existence is carried out in Section~\ref{secGlobal}. To this end we use a bootstrap argument that extends the methods of proof developed in~\cite{jesus1} and~\cite{jesus2}: in a nutshell, the nonlinear structure~\eqref{nonLinForm} takes advantage of the integrability of $1/a(t)$ in order to achieve balanced commutator estimates for derivatives in $L^2$ and $L^{\infty}$; this allows us to close the bootstrap. Finally, in Section~\ref{secDecay} we establish the sharp future decay estimates. The simple proof presented here is a variation on an idea by Pedro Girão~\cite{girao} to deal with the de Sitter case, which was previously implemented in~\cite{amol}. Here we present a streamlined and extended version of this strategy; streamlined by avoiding the need to use conformal time and extended to semilinear equations and the entire FLRW family under consideration.   

\section{Setup and Main Results}
\label{setup}

Let $(\mathcal{M},{\mathbf{g}})$
be a Friedman-Lemaitre-Robertson-Walker (FLRW) spacetime
with topology $\mathbb{R}^+\times\mathbb{R}^n$ and metric
\begin{equation}
\label{metricFLRW}
    {\metric} := -dt^2 + a^2(t)\sigma_{ij}dx^idx^j\;,  
\end{equation}
where $\sigma_{ij}:=\sigma_{ij}(x)$ are the components of a Riemannian metric in $\mathbb{R}^n$.  We will consider cosmologies undergoing an accelerated expansion in the direction of positive time $t$: expansion corresponds to 
\begin{equation}
 \dot a:=\partial_t a>0\;,
\end{equation}
 and the accelerated character of this expansion can by codified by imposing the integrability condition 
\begin{equation}
\label{intCond}
 \int_{t_0}^{\infty}  \frac{1}{a(s)} ds<\infty\;, 
\end{equation}
where, from now on,  $t_0>0$ is fixed. We also assume that $a(t)>0$ for all $t\geq t_0$.

Consider the covariant wave operator, defined by 
\begin{equation}
\square_{\metric}\phi = \frac{1}{\sqrt{|\metric|}}\partial_{\alpha}\left(\metric^{\alpha\beta} \sqrt{|\metric|} \partial_{\beta}\phi\right ) \;,
\end{equation}
where $|\metric|=-\det(\metric_{\alpha\beta})$, $\metric^{\alpha\beta}$ are the components of the inverse of $\metric_{\alpha\beta}$, and where, as usual, greek indices run from $0$ to $n$. For the FLRW metric~\eqref{metricFLRW},we have
\bea
\lb{flrw_eq}
\square_{\metric}\phi=-\partial_t^2 \phi-n\frac{\partial_t {a}}{a}\partial_t \phi+ \frac{1}{a^2}\Delta_{\sigma} \phi\;,
\eea
where $\Delta_{\sigma}$ is the Laplace operator of the metric $\sigma$, defined by 
\begin{equation}
\label{DeltaInComp}
\Delta_{\sigma}\phi = \frac{1}{\sqrt{|\sigma|}}\partial_{i}\left(\sigma^{ij} \sqrt{|\sigma |} \partial_{j}\phi\right ) \;,
\end{equation}
for $|\sigma|=\det(\sigma_{ij})$, and where $\sigma^{ij}$ are the components of the inverse of $\sigma_{ij} $,  with latin indices taking values in the range $1$ to $n$.

Let $1\leq A,B,C\leq d$ and $\phi^A:\mathbb{R}^{1+n}\rightarrow \mathbb{R}$. In this work we study solutions
to the Cauchy problem for systems of
semilinear wave equations of the form 
\begin{equation}
\begin{cases}
\square_{\metric}\phi^A=\tilde \nN^{A,\alpha \beta}_{BC}(\phi)\partial_{\alpha} \phi^B \partial_{\beta} \phi^C\;, \\
\phi^{A}(t_0,x)=\phi^{A}_0(x) \;\;, \; \partial_t\phi^{A}(t_0,x)=\phi^{A}_1(x) \;,
\end{cases}
\end{equation}
for $(\phi^{A}_0,\phi^{A}_1)\in H^{K+1}(\mathbb{R}^{n})\times H^{K}(\mathbb{R}^{n})$, $K\geq  n +1$,
where the nonlinearities take the form 
    \begin{align} \label{nonLinForm0}
	&\tilde \nN^{A,0 0}_{BC}=\nN^{A,0 0}_{BC}\;,
          \\
        &\tilde \nN^{A,0 j}_{BC}=a^{-1}(t)\,\nN^{A,0 j}_{BC}\;,
        \\  
        &\tilde \nN^{A,i 0}_{BC}=a^{-1}(t)\,\nN^{A,i 0}_{BC}\;,
        \\  
        &\tilde \nN^{A,i j}_{BC}=a^{-2}(t)\,\nN^{A,i j}_{BC}\;,
	\end{align}
with $\nN^{A,\alpha \beta}_{BC}\in C_{b}^{\infty}(\mathbb{R}^{d})$, i.e. the functions $\nN$ have uniformly bounded derivatives of all orders.  Note that by using the Kronecker symbol we can compress the form of the nonlinearities to a single expression by writing 
    \begin{align} \label{nonLinForm}
	&\tilde \nN^{A,\alpha \beta}_{BC}=a^{-2+\delta_{0\alpha}+\delta_{0\beta}}\nN^{A,\alpha \beta}_{BC}\;.
	\end{align}

The main results of our paper are compiled in

\begin{theorem}
\label{main}
Let $K\geq n+1$, where $n\geq 2$ is the spatial dimension of a FLRW spacetime, with topology $\mathbb{R}^{+}\times\mathbb{R}^n$ and smooth metric of the form~\eqref{metricFLRW}, whose spatial geometry satisfies
\begin{equation}
\label{smallSigmaK}
\sum_{i,j=1}^n \left(\lp{\sigma_{ij}-\delta_{ij}}{L^{\infty}(\bbR^n)}  +\sum_{k=1}^{K+1} \lp{\partial_x^k\sigma_{ij}}{L^{\infty}(\bbR^n)}\right)  =: C_{\sigma}<\infty \;.
\end{equation}
 Consider initial data  $\phi^{A}_0,\phi^{A}_1:\bbR^n\rightarrow\mathbb R$, $1\leq A\leq d$,  such that, for a fixed $K\geq n+1$,
\begin{equation}
\label{smallInitial data}
\sum_{1\leq A\leq d}\left(\lp{\phi^{A}_0}{H^{K+1}(\bbR^n)}  +\lp{\phi^{A}_1}{H^{K}(\bbR^n)}\right)=:C_0<\infty\;.
\end{equation}
Then, given $t_0>0$, there exists $\delta_0>0$, such that, if $C_{\sigma}+C_0\leq \delta_0$, the initial value problem 
\begin{equation}
\label{systemMain0}
\begin{cases}
\square_{\metric}\phi^A=a^{-2+\delta_{0\alpha}+\delta_{0\beta}}\nN^{A,\alpha \beta}_{BC}(\phi) \partial_{\alpha} \phi^B \partial_{\beta} \phi^C\;, \\
\displaystyle
\phi^{A}(t_0,x)=\phi^{A}_0(x) \;\;, \; \partial_t\phi^{A}(t_0,x)=\phi^{A}_1(x) \;,
\end{cases}
\end{equation}
with $\nN^{A,\alpha \beta}_{BC}\in C_{b}^{\infty}(\mathbb{R}^{d})$, admits a unique solution 
$(\phi^{A},\partial_t\phi^{A})\in L^{\infty}([t_0,T), H^{K+1}(\bbR^n))\times L^{\infty}([t_0,T), H^{K}(\bbR^n)) $\;.

Concerning the asymptotic behavior of the solutions, given a fixed $0\leq k<K-\frac{n}{2}$, we highlight that:  
\begin{enumerate}
\item 
For a general expanding factor we have
\bea 
\label{mainDecayEst}
\lp{\partial_t\partial_x^k\phi^{A} (t,\,\cdot\,)}{L^{\infty}}  \lesssim C_0 \left( \int_{t_0}^t  a^{n-2}(s) ds \right) a^{-n}(t) \;.
\eea
\begin{enumerate}
\item in the case of a  power law expansion $a(t)=t^p$, $p>1$,~\eqref{mainDecayEst} becomes 
\bea 
\lp{\partial_t\partial_x^k\phi^{A} (t,\,\cdot\,)}{L^{\infty}}\lesssim C_0 t^{-2p+1}\;,
\eea
\item and in the de Sitter case $a(t)=e^{Ht}$, $H>0$,~\eqref{mainDecayEst} reads 
\bea 
\lp{\partial_t\partial_x^k\phi^{A} (t,\,\cdot\,)}{L^{\infty}}\lesssim C_0 e^{-2Ht}\;.
\eea
\end{enumerate}
\item Moreover, for a general expanding factor there exists a function $\phi_{\infty}=(\phi_{\infty}^A):\bbR^n\rightarrow \bbR^d$, such that  we have
\begin{equation}
\lp{\partial_x^k\left(\phi^{A}(t,\,\cdot\,)-\phi^{A}_{\infty}\right)}{L^{\infty}}\rightarrow 0\;,
\label{limit_function}
\end{equation}
as $t\rightarrow \infty$.

\begin{enumerate}
\item in the case of  a power law expansion $a(t)=t^p$, $p>1$, we have
\bea 
\label{decay2function1}
\lp{\partial_x^k\left(\phi^{A}(t,\,\cdot\,)-\phi^{A}_{\infty}\right)}{L^{\infty}}\lesssim C_0 t^{-2p+2}\;,
\eea
\item and in the de Sitter case $a(t)=e^{Ht}$, $H>0$, we get
\bea 
\label{decay2function2}
\lp{\partial_x^k\left(\phi^{A}(t,\,\cdot\,)-\phi^{A}_{\infty}\right)}{L^{\infty}}\lesssim C_0 e^{-2Ht}\;.
\eea
\end{enumerate}
\end{enumerate}
\end{theorem}

\vspace{1cm}

\begin{remark}
\label{rmkExs}
 The previous result applies to the following particular cases:
\begin{itemize}
\item Let $\phi$ be a wave map with base manifold one of our FLRW spacetimes $({\mathcal M}, g)$ and target manifold a given Riemannian manifold $({\mathcal N},h)$ with~{\em uniformly bounded geometry}~\cite[Chapter 6]{taoDispersive}, by which we mean that we can cover $\mathcal N$ with coordinate charts of radius bounded from below, where the Christoffel symbols of $h$, that we denote by $\Gamma^A_{BC}$, are bounded and have bounded derivatives of all orders.  Then, in the small data setting, the wave map equations take the form 
\begin{equation}
\square_{\metric}\phi^A=-{\metric}^{\alpha \beta}\Gamma^{A}_{BC}(\phi) \partial_{\alpha} \phi^B \partial_{\beta} \phi^C\;,
\end{equation}
which clearly fits into our framework. To this effect recall that $g^{\alpha \beta}$ are the components of the inverse metric which is given by
\begin{equation}
\label{metricInvFLRW}
    {\metric}^{-1} := -\partial_t\otimes\partial_t + a^{-2}(t)\sigma^{ij}\partial_{x^i}\otimes \partial_{x^j}\;,  
\end{equation}
where $\sigma^{ij}$ are the components of $\sigma^{-1}$.
 \item  Arguably  the most famous example of Fritz John's ``blow up'' equations~\cite{john_blow} is 
\begin{equation}
\label{John}
    \square_{\metric}\phi= (\partial_{t} \phi)^2\;.
\end{equation}
Recall that in 1+3 dimensional Minkowski spacetime, i.e., if $g=\eta$, where $\eta$ is the flat metric, all solutions arising from arbitrarily small, but non-trivial, smooth and compactly supported initial data, blow up in finite time. 
However, if $g$ is one of our FLRW metrics then the equation fits into our framework and as a consequence our results show that, in such case, the equation satisfies small data global existence to the future; so we see that the (small data) finite time blow up to the future disappears as a consequence of the accelerated expansion.
\item Linblad-Rodnianski's basic example of a system that does not satisfy the null condition but satisfies the weak-null condition (see \cite{lind_rod} for more information) formally generalizes to our setting to yield
\begin{equation}
\begin{cases}
 \square_{\metric}\phi_1= 0\;,
\\
    \square_{\metric}\phi_2= (\partial_{t} \phi_1)^2\;.
\end{cases}
\end{equation}
Contrary to what happens in Minkowski, in the FLRW case small data global existence for this system is not as surprising in view of the fact that we also have global existence for Fritz John's equation~\eqref{John}.  
\end{itemize}
\end{remark}

%
%

\section{Energy Formalism}
\label{secEnergy}

We define the {\it{Energy-Momentum Tensor}} to be
\begin{equation}
    T_{\alpha\beta} = \partial_\alpha \phi\, \partial_\beta \phi
    - \frac{1}{2}\metric_{\alpha\beta}
    \partial^\mu \phi\,  \partial_\mu \phi  \;.
\end{equation}

Let $D$ be the Levi-Civita connection of the metric $\metric$. The divergence of the energy momentum is then
\[D^\alpha T_{\alpha\beta}=  \partial_{\beta}\phi\square_{\metric}\phi\; .\] 
Given a (smooth) 
vector field $X$, we define the 1-form
\[{}^{(X)}P_\alpha = T_{\alpha\beta}X^\beta .\]
Taking its divergence yields
\begin{equation} \label{divergence}
    D^\alpha {}^{(X)}P_\alpha = \frac{1}{2}\;^{(X)}\pi^{\alpha\beta}T_{\alpha\beta}+X\square_{\metric}\phi\;,
\end{equation}
where 
\[^{(X)}\pi_{\alpha\beta}:= {\mathcal L}_X \metric _{\alpha\beta} = D_{\alpha}X_{\beta}+D_{\beta}X_{\alpha} \;,\] 
is a symmetric 2-tensor known as the \textit{Deformation Tensor}  
 of
$\metric$ with respect to $X$. Integrating the divergence identity~\eqref{divergence} over the time slab
\[\{(t,x) \;|\;t_0\leq t\leq t_1\}\] and using Stokes'
theorem we get the following {\it{Multiplier Identity}}
\begin{equation} \label{multident}
\begin{split}
    \int_{t=t_0}\; ^{(X)}P_\alpha N^\alpha |\metric|^{\frac{1}{2}}dx -
    \int_{t=t_1}\; ^{(X)}P_\alpha N^\alpha |\metric|^{\frac{1}{2}}dx 
    = \int_{t_0}^{t_1}\int_{\mathbb{R}^n}\left(\frac{1}{2} \;^{(X)}\pi^{\alpha\beta}T_{\alpha\beta}+
    X\phi \cdot \square_{\metric}\phi\right)|\metric|^{\frac{1}{2}}dxdt\;,
\end{split}
\end{equation}
where $N=\partial_t$ is the future pointing unit normal to the time slices $t=const$, $|\metric|=-\det(\metric_{\alpha\beta})=a^n\det(\sigma)=a^n|\sigma|$, with $|\sigma|:=\det(\sigma_{ij})$, and $dx=dx^1\cdots dx^n$. 
The integrand $^{(X)}P_\alpha N^\alpha$ in \eqref{multident} is the {\it{Energy Density}}
associated to $X$. 

We can also control the sign of the contraction of the deformation tensor with the energy-momentum  tensor, on the right hand side, by choosing an appropriate multiplier vector field $X$. In fact we have

\begin{lemma}
For $X=a^l \partial_t$, where $a$ is the expanding factor and $l\in\mathbb{R}$, we have
\begin{equation}
\;^{(X)}\pi^{\alpha\beta}T_{\alpha\beta}=(n-l)a^{l-1}\dot a \left(\partial_t \phi\right)^2 + (2-n-l)a^{l-3}\dot a \,\sigma^{ij} \partial_i\phi \partial_j \phi\;,\label{def_tensor_explicit}
\end{equation}
where $\sigma^{ij}$ is the inverse of $\sigma_{ij}$.
\end{lemma}   
\begin{proof}
To compute the deformation tensor we start by noting that ${\mathcal L}_X \sigma =0$,
\begin{eqnarray*}
 {\mathcal L}_X dt = d \iota_X dt = d(a^l) = l a^{l-1}\dot a dt\;,
\end{eqnarray*}
and that
\begin{eqnarray*}
 {\mathcal L}_X a = a^l\dot a\;,
\end{eqnarray*}
in order to compute
\begin{eqnarray*}
\;^{(X)}\pi &=& {\mathcal L}_X \metric = {\mathcal L}_X(-dt^2 + a^2 \sigma_{ij}dx^idx^j) 
\\
&=& -2 dt {\mathcal L}_X dt +  2 a^{l+1} \dot a \sigma_{ij}dx^idx^j
\\
&=& -2l a^{l-1}\dot a dt^2 +  2 a^{l+1} \dot a \sigma_{ij}dx^idx^j\;.
\end{eqnarray*}
It then follows that 
\begin{eqnarray*}
\;^{(X)}\pi_{\alpha\beta}T^{\alpha\beta} &=& -2 l a^{l-1} \dot a\left( \left(\partial_t \phi\right)^2  + \frac12  \partial_{\alpha}\phi\partial^{\alpha} \phi \right)
\\
&& + 2 a^{l+1} \dot a \sigma_{ij} \left( \partial^i \phi\partial^j \phi -\frac12 g^{ij}  \partial_{\alpha}\phi\partial^{\alpha} \phi \right)\;,
\end{eqnarray*}
and the desired result is then a consequence of the identities
\begin{equation}
 \partial_{\alpha}\phi\partial^{\alpha}\phi = - \left(\partial_t \phi\right)^2 +a^{-2} \sigma^{ij}\partial_i\phi\partial_j\phi\;,
\end{equation}
and 
\begin{equation}
\sigma_{ij}\partial^i\phi\partial^j\phi=a^{-4}\sigma^{ij}\partial_i\phi\partial_j\phi\;.
\end{equation}
\end{proof}

We thus choose $l=2-n$, for which we have   
\begin{equation}
\label{goodSign}
^{(X)}\pi_{\alpha\beta}T^{\alpha\beta}\geq 0\;,
\end{equation}
and define
\begin{equation}
\label{energyDef}
E[\phi](t):= \int_{\{t\}\times\bbR^n}\; ^{(X)}P_\alpha N^\alpha |\metric|^{\frac{1}{2}}dx= \frac12\int_{\mathbb{R}^n}\left[a^{2}(\partial_{t}\phi)^{2}+\sigma^{ij}\partial_{i}\phi\partial_{j}\phi\right](t,x)  |\sigma|^{\frac12}dx\;.
\end{equation}
Consequently, the identities~\eqref{multident} and~\eqref{def_tensor_explicit} give rise to  
\begin{equation}
\label{energyEst0}
E(t_1)\leq E(t_0) + \int_{t_0}^{t_1}\int_{\mathbb{R}^n}
 a^2 \left | \partial_t\phi \, \square_{\metric}\phi \right| \,|\sigma|^{\frac{1}{2}}dxdt\;.
\end{equation}

Introducing  the $L^{2}$ norm  defined by
\[\lp{f}{\Lx}=\left(\int_{\mathbb{R}^{n}}\left|f(x)\right|^2 |\sigma|^{\frac{1}{2}}dx\right)^{\frac12}\;, \]
we can apply the Cauchy-Schwarz inequality followed by Young's inequality to~\eqref{energyEst0} to obtain
\begin{eqnarray*}
E(t)&\leq& E(t_0) + \int_{t_0}^{t_1} \lp{a\, \partial_t \phi}{\Lx} \|a\, \square_{\metric}\phi \|_{\Lx} dt
\\
 &\leq& E(t_0) + \sqrt{2}\sup_{t_0\leq t\leq t_1 } \left\{E^{1/2}(t)\right\} \int_{t_0}^{t_1}  \|a\, \square_{\metric}\phi \|_{\Lx} dt
 \\
&\leq&  E(t_0) + \frac{\epsilon}{2} \sup_{t_0\leq t\leq t_1 } E(t) + \frac{1}{\epsilon} \left(\int_{t_0}^{t_1}  \|a\, \square_{\metric}\phi \|_{\Lx} dt\right)^2\;,
\end{eqnarray*}
where $\epsilon$ can be any positive constant, which when chosen sufficiently small allows us to conclude that  
\begin{eqnarray*}
\sup_{t_0\leq t\leq t_1 } E(t)\leq \frac{1}{1-\frac{\epsilon}{2}} \left(E(t_0) + \frac{1}{\epsilon} \left(\int_{t_0}^{t_1}  \|a\, \square_{\metric}\phi \|_{\Lx} dt\right)^2 \right)\;. 
\end{eqnarray*}
%
From the previous we immediately obtain our main energy estimate:  
\begin{theorem}
 Let $\metric$ be the FLRW metric~\eqref{metricFLRW}. Then, there exists a constant $C_1>0$ such that given $\phi: [t_0,T) \times \bbR^{n}\rightarrow \bbR$, the following energy estimate holds
 \begin{equation} 
 \label{energyEst}
 \sup_{t_0\leq t< T}E^{1/2}(t) \leq C_1\Big(E^{1/2}(t_0)+\int_{t_0}^{T}a(t)\lp{ \square_{\metric}\phi }{\Lx} dt \Big)\;, 
  \end{equation}
for the energy $E=E[\phi]$ defined in~\eqref{energyDef}.
\end{theorem}


\section{local existence}
\label{secLocal} 

Consider two conformally related metrics in $\bbR^{1+n}$
$${\metric}=\Omega^2 \tilde\metric\;.$$
Then a direct computation shows that 
\begin{equation}
\label{waveConformal}
\square_{\metric}\phi = \Omega^{-2}\square_{\tilde\metric}\phi+(n-1)\Omega^{-3}\tilde\metric^{\alpha\beta} \partial_{\alpha}\Omega\partial_{\beta}\phi\;.
\end{equation}

If we let $\metric$ be the FLRW metric~\eqref{metricFLRW} and consider the standard change of time variable
\begin{equation}
\tau=\int_{t_0}^t \frac{1}{a(s)} ds \;,
\end{equation}
then
\begin{equation}
\metric=a^2(\tau)\left( -d\tau^2+\sigma_{ij} dx^idx^j\right)=:a^2 \eta_{\sigma}\;.
\end{equation}
Using~\eqref{waveConformal} with $\Omega=a$ and $\tilde \metric= \eta_{\sigma}$ we conclude that the semilinear equation 
$$\square_{\metric}\phi=F(\phi,\partial \phi)$$ 
is equivalent to 
$$\square_{\eta_{\sigma}}\phi=(n-1)\frac{\partial_{\tau}a}{a}\partial_{\tau}\phi+a^2F(\phi,\partial \phi)\;.$$ 
We can then apply classical local well-posedness results for nonlinear wave equations on perturbations of Minkowski (see for instance~\cite[Theorem 4.1]{sogge} and~\cite[Theorem 6.1 and Theorem 6.6]{lukLectures}, which when translated back to our original setting lead to  

\begin{theorem}
\label{local}
Let $(\mathcal{M},\metric)$ be a FLRW spacetime, with $\mathcal{M}=\mathbb{R}^+\times\mathbb{R}^n$ and smooth metric of the form~\eqref{metricFLRW}. 
Assume moreover that  
\begin{equation}\label{inverse_metric_small}
\sup_{x\in\bbR^n}\sum_{i,j} |\sigma^{ij}-\delta^{ij}|<\frac{1}{10}\;. 
\end{equation}
Then, the initial value problem 
\begin{equation}
\begin{cases}
\square_{\metric}\phi=F(\phi,\partial \phi) \;, \\
\phi(t_0,x)=\phi_0(x) \;\;, \; \partial_t\phi(t_0,x)=\phi_1(x) \;,
\end{cases}
\end{equation}
with $F\in C^{\infty}$, $F(0,0)=0$, and $(\phi_0,\phi_1)\in H^{K+1}(\bbR^n)\times H^K(\bbR^n)$, $K\geq n+1$, admits a unique solution $(\phi,\partial_t\phi)\in L^{\infty}([t_0,T), H^{K+1}(\bbR^n))\times L^{\infty}([t_0,T), H^{K}(\bbR^n)) $, where 
$T=T(\lp{(\phi_0,\phi_1)}{H^{K+1}\times H^K})>0$\;. 
\end{theorem}

\section{Global Existence for small data}
\label{secGlobal} 

In this section we establish small data global existence for the system~\eqref{systemMain0} under the conditions of Theorem~\ref{local}. To do that recall that $C_0$ denotes the size of the initial data~\eqref{smallInitial data}. We
start by defining
\begin{equation*}
\mathcal{E}_{K}(t):=\sum_{1\leq A\leq d}\sum_{k=0}^K E^{1/2}[\partial_x^{k}\phi^A](t)
=\sum_{1\leq A\leq d}\frac{1}{\sqrt{2}}\sum_{k=0}^K\left(\int_{\{t\}\times \bbR^n} a^2 (\partial_t\partial_x^{k}\phi^A)^2 + \sigma^{ij} \partial_i\partial_x^{k}\phi^A\partial_j\partial_x^k\phi^A
\right)^{\frac12} \;.
\end{equation*}

Some comments concerning notation are in order: first we are using $\partial_x^k$ to denote any differentiation of order $k$ with respect to the spatial variables $x^i$, i.e., any differentiation of the form 
$$\frac{\partial^{k_1}}{\partial {x^{i_1}}}\cdots \frac{\partial^{k_l}}{\partial {x^{i_l}}}\;,$$ 
 with $\sum k_i = k$; secondly, we are omitting the volume form $|\sigma|^{\frac12}dx$ in order to not overburden the notion; this is clearly not an issue since  by the smallness condition on the metric we have $|\sigma|\sim 1$, uniformly on $x$, which allows us to drop this coefficient from all spatial $L^2$ based norms.  On this note it might be helpful to make it clear that by choosing $C_{\sigma}$ sufficiently small, there exits $C>0$ such that  
 \begin{equation}
 \label{simMetrics}
 C^{-1}\delta^{ij} \xi_i\xi_j \leq \sigma^{ij} \xi_i\xi_j \leq C \delta^{ij} \xi_i\xi_j\;, 
 \end{equation}
 %
for all $\xi=(\xi_i)\in\bbR^n\;.$

We proceed by a continuity argument. Let $K\geq n+1$ and
 assume that $M$ is a large enough constant so that
\begin{equation}\label{initial_cond_boot}
	\mathcal{E}_K(t_0)\leq \frac{MC_0}{C_2} ,
\end{equation}
where $C_2$ is a constant, independent of $C_0$ and $C_{\sigma}$, that will be specified in the course of the proof. Next we assume as  bootstrap condition that  $T$ is the supremum over all times of existence $t\geq t_0$ for which  
\begin{equation}
\label{bootAss}
\sup_{t_0\leq t<T}\mathcal{E}_{K}(t) \leq 4MC_0\;.
\end{equation}
That a $T>t_0$ in such conditions exists is a direct consequence of Theorem~\ref{local}.

Using Sobolev embedding, the relation~\eqref{simMetrics}, and the bootstrap assumption  we see that for $0\leq l < K -\frac{n}{2}$, and $0\leq t<T$, we have
\begin{equation}
\label{dxBound}
\lp{\partial_x^l\partial_x\phi^A(t,\,\cdot\,)}{L^{\infty}(\bbR^n)} \leq C \lp{\partial_x\phi^A(t,\,\cdot\,)}{H^K(\bbR^n)}\leq C \mathcal{E}_{K}(t) \leq  4MCC_0\; , 
\end{equation}
for all $A\in \{1,2,...,d\}$. Moreover, for  $0\leq l < K -\frac{n}{2}$, and $0\leq t<T$, we get  

\begin{equation}
\label{dtBound}
\lp{\partial_x^l\partial_t\phi^A(t,\,\cdot\,)}{L^{\infty}(\bbR^n)} \leq C \lp{\partial_t\phi^A(t,\,\cdot\,)}{H^K(\bbR^n)}\leq C a^{-1}(t) \mathcal{E}_{K} (t)\leq 4MCC_0\,a^{-1}(t)\; ,
\end{equation}
again for all $A\in \{1,2,...,d\}$. Let $0\leq k\leq K$. The spatial derivatives satisfy the equation
\begin{eqnarray*}
\square_{\metric}\partial_x^k\phi^A&=&\partial_x^k\square_{\metric}\phi^A + [\square_{\metric},\partial_x^k] \phi^A\\ 
&=& \partial_x^k \left(\tilde \nN^{A,\alpha \beta}_{BC}(\phi)\partial_{\alpha} \phi^B \partial_{\beta} \phi^C \right) +[\square_{\metric},\partial_x^k]\phi^A\;,
\end{eqnarray*}
so that applying the energy estimate~\eqref{energyEst} for $A\in\{1,2,...,d\}$, summing and taking the supremum gives 
 \begin{eqnarray} 
 \label{energyEst2}
 \sup_{t_0\leq t<T}\mathcal{E}_{K}(t) 
 &\leq& C
  \mathcal{E}_{K}(t_0) +
 C \sum_{A=1}^d\sum_{k=0}^K
 \left\{
 \int_{t_0}^{T}a(t)\lp{ \partial_x^k \left(\tilde \nN^{A,\alpha \beta}_{BC}(\phi)\partial_{\alpha} \phi^B \partial_{\beta} \phi^C \right)  }{\Lx} dt \,+
 \right.
 \\
 \nonumber
 &&+
\left. \int_{t_0}^{T}a(t)\lp{[\square_{\metric},\partial_x^k]\phi^A }{\Lx} dt \right\}  \;.
  \end{eqnarray}

Let us first concentrate on the second term on the right hand side of the previous energy estimate. Using the form of the non-linearities~\eqref{nonLinForm}  we get
 \begin{equation} 
 \label{Nterms}
\lp{ \partial_x^k \left(\tilde \nN^{A,\alpha \beta}_{BC}(\phi)\partial_{\alpha} \phi^B \partial_{\beta} \phi^C \right)  }{\Lx} 
\leq C
\left( a(t) \right)^{\delta_{0\alpha}+\delta_{0\beta}-2} \lp{ \partial_x^k \left(\nN^{A,\alpha \beta}_{BC}(\phi)\partial_{\alpha} \phi^B \partial_{\beta} \phi^C \right)  }{L^2(\bbR^n)}\;. 
 \end{equation}

To simplify notation we will use $\nN$ to collectively denote all the functions $\nN^{A,\alpha \beta}_{BC}$. That being said, we note that modulo some multiplicative positive constants arising from the application of Leibniz rule, the terms 
$\lp{ \partial_x^k \left(\nN(\phi)\partial_{\alpha} \phi^B \partial_{\beta} \phi^C \right)  }{L^2(\bbR^n)}$ can be bounded by sums of terms of the form 
$$\lp{ \partial_x^{k_1} \nN(\phi)\partial_x^{k_2} \left(\partial_{\alpha} \phi^B \partial_{\beta} \phi^C \right)  }{L^2(\bbR^n)}\;,$$
with $k_1,k_2\geq0$ and $k_1+k_2=k$. 
In particular, either $k_1\leq k/2$ or $k_2\leq k/2$. 

Let us start with the case $k_1\leq k/2$. Recall that, by assumption, 
 \begin{equation} 
 \label{NBound}
 \lp{\partial_{\phi}^l \nN}{L^{\infty}(\bbR^d)}\leq C\;,
  \end{equation}
 for all $0\leq l \leq K$. This gives the necessary control $\lp{ \nN(\phi(t,\,\cdot\,))}{L^{\infty}(\bbR^n)}\leq C$ needed in the case  $k_1=0$. If $k_1\geq 1$, the chain rule implies that   $\lp{\partial_x^{k_1} (\nN(\phi))}{L^{\infty}(\bbR^n)}$ is bounded by sums of terms of the form 
$$O(1)\Pi_{s_i} \lp{\partial_{x}^{s_i}\phi^A}{L^{\infty}(\bbR^n)}\;,$$
with $\sum s_i=k_1$, and $s_i\geq 1$. In view of~\eqref{dxBound} we see that these terms are bounded, since we have $0\leq l=s_i-1\leq k_1-1\leq k/2-1\leq K/2-1<K-n/2$, where the last inequality follows from the fact that $K\geq n+1>n-2$. Consequently 
 \begin{equation} 
 \label{Nterms1}
 \lp{ \partial_x^{k_1} \nN(\phi)\partial_x^{k_2} \left(\partial_{\alpha} \phi^B \partial_{\beta} \phi^C \right)  }{L^2(\bbR^n)}
 \leq C
  \lp{ \partial_x^{k_2} \left(\partial_{\alpha} \phi^B \partial_{\beta} \phi^C \right)  }{L^2(\bbR^n)}\;.
   \end{equation}
But now the right hand side is bounded by sums of terms of form 
$$\lp{ \partial_x^{\tilde k_1} \partial_{\alpha} \phi^B  \partial_x^{\tilde k_2}\partial_{\beta} \phi^C  }{L^2(\bbR^n)}\;,$$
with $\tilde k_1+\tilde k_2=k_2\leq k$. We may then assume without loss of generality that $\tilde k_1\leq k/2$. Since $K\geq n+1$ we have $0\leq \tilde k_1\leq k/2\leq K/2<K -n/2$ and therefore we are allowed to use either~\eqref{dxBound}, if $\alpha\neq 0$, or~\eqref{dtBound}, if $\alpha = 0$, to conclude that 
 \begin{equation} 
\lp{ \partial_x^{\tilde k_1} \partial_{\alpha} \phi^B }{L^{\infty}(\bbR^n)}
\leq \left( a(t) \right)^{-\delta_{0\alpha}}4MCC_0\;. 
\end{equation}
But then we get 
 \begin{eqnarray} 
  \label{Nterms2}
\lp{ \partial_x^{\tilde k_1} \partial_{\alpha} \phi^B  \partial_x^{\tilde k_2}\partial_{\beta} \phi^C }{L^2(\bbR^n)}
&\leq&
 \lp{ \partial_x^{\tilde k_1} \partial_{\alpha} \phi^B }{L^{\infty}(\bbR^n)}
\lp{ \partial_x^{\tilde k_2}\partial_{\beta} \phi^C  }{L^2(\bbR^n)}
\\
\nonumber
&\leq& \left( a(t) \right)^{-\delta_{0\alpha}}4MCC_0 \left( a(t) \right)^{-\delta_{0\beta}} \mathcal{E}_K(t)\;. 
\end{eqnarray}
Using~\eqref{Nterms1} and~\eqref{Nterms2} we finally establish 
\begin{equation}
\label{mainEstks}
\lp{ \partial_x^{k_1} \nN(\phi)\partial_x^{k_2} \left(\partial_{\alpha} \phi^B \partial_{\beta} \phi^C \right)  }{L^2(\bbR^n)}\leq \left( a(t) \right)^{-\delta_{0\alpha}-\delta_{0\beta}}4MCC_0 \mathcal{E}_{K}(t) \;,
\end{equation}
 provided $k_1\leq k/2$. 
 
Let us now consider the case $k_1>k/2$. In such case $k_2< k/2$ and therefore we see that  $$\lp{ \partial_x^{k_2} \left(\partial_{\alpha} \phi^B \partial_{\beta} \phi^C \right)}{L^{\infty}(\bbR^n)}$$ is bounded by terms of the form 
 \begin{equation} 
 \label{dNinfty}
 \lp{\partial_x^{\hat k_1} \partial_{\alpha}\phi^B }{L^{\infty}(\bbR^n)}  \lp{\partial_x^{\hat k_2}\partial_{\beta} \phi^C }{L^{\infty}(\bbR^n)}
  \leq \left( a(t) \right)^{-\delta_{0\alpha}-\delta_{0\beta}}4MCC_0 \mathcal{E}_{K}(t)\;,
\end{equation}
where the last estimate is a consequence of~\eqref{dxBound} and~\eqref{dtBound} and the fact that $\hat k_1+\hat k_2= k_2<k/2\leq K/2$ and $K\geq n+1$.

Next we need to control 
 \begin{equation} 
 \lp{\partial_x^{ k_1} \nN(\phi) }{L^2(\bbR^n)} \;\;,\; k_1>k/2\;.
 \end{equation}
Applying the chain rule and using~\eqref{NBound}, we see that $|\partial_x^{ k_1} \nN(\phi)|$ is controlled by sums of terms of the form 
$$|O(1)\Pi_{s_i} \partial_{x}^{s_i}\phi^A|\;\;,\; \text{ with } \sum s_i=k_1\;.$$
By an appropriate relabeling, set $s_1=\max\{s_i\}$ so that $s_i\leq k_1/2\leq k/2$, for all $i\neq 1$, which, in view of~\eqref{dxBound},  implies  
$$ \lp{\partial_{x}^{s_i}\phi^A}{L^{\infty}(\bbR^n)}\leq 4MCC_0\;\;,\; \text{ for all } i\neq 1\;.$$
Consequently, using the bootstrap assumption,
$$\lp{\Pi_{s_i} \partial_{x}^{s_i}\phi^A}{{L^2(\bbR^n)}}\leq 4MCC_0 \lp{ \partial_{x}^{s_1}\phi^A}{{L^2(\bbR^n)}}\leq (4MCC_0)^2\;,$$
from which, by decreasing $C_0$ if necessary, we can establish 
 \begin{equation} 
 \lp{\partial_x^{ k_1} \nN(\phi) }{L^2(\bbR^n)} \leq 1\;.
 \end{equation}
The last estimate together with~\eqref{dNinfty} then allows us to conclude that~\eqref{mainEstks} also holds in the case $k_1>k/2$.  So, for any $0\leq k\leq K$, using~\eqref{Nterms} leads to 
 \begin{equation} 
 \label{NtermsFinal}
\lp{ \partial_x^k \left(\tilde \nN^{A,\alpha \beta}_{BC}(\phi)\partial_{\alpha} \phi^B \partial_{\beta} \phi^C \right)  }{\Lx} 
\leq
a^{\delta_{0\alpha}+\delta_{0\beta}-2}a^{-\delta_{0\alpha}-\delta_{0\beta}}4MCC_0 \mathcal{E}_{K}(t)=a^{-2}4MCC_0 \mathcal{E}_{K}(t)\;, 
 \end{equation}
and 
 \begin{equation} 
 \label{nonlinEst}
\int_{t_0}^{T}a(t)\lp{ \partial_x^k \left(\tilde \nN^{A,\alpha \beta}_{BC}(\phi)\partial_{\alpha} \phi^B \partial_{\beta} \phi^C \right)  }{\Lx} dt \leq 4MC C_0 \int_{t_0}^{T}\frac{\mathcal{E}_{K}(t)}{a(t)} dt \;.
 \end{equation}

 We now consider the last term in~\eqref{energyEst2}. In the case $k=K$, the commutator  $[\square_{\metric},\partial_x^k]$ is the difference of two differential operators of order $K+2$ and this is worrisome since, a priori, our bootstrap assumption only gives control of derivatives up to order $K+1$! 
 But it is well known  that the top derivatives cancel out (see for instance ~\cite[Section 6.2]{alinhac}  or~\cite{m_lec}).
 For the sake of completeness, we show here that $[\square_{\metric},\partial_x^k]$ is in fact of order $k+1$, and that moreover enough factors involving the spatial metric $\sigma$ appear and provide, via the $K$th order near flatness condition~\eqref{smallSigmaK}, a small parameter $C_{\sigma}$ that will allow us to close our bootstrap argument.   

Using~\eqref{flrw_eq} we see that 
\begin{equation}
[\square_{\metric},\partial_x^k] = a^{-2} [\Delta_{\sigma},\partial_x^k]\;.
\end{equation}
 Then, if we note that Leibniz rule can be written as 
\begin{equation}
\partial_x^k(fg)=\sum_{k_1+k_2=k} c_{k_1,k_2}\partial_x^{k_1}f\partial_x^{k_2}g\;, 
\end{equation}
with the $c_{k_1,k_2}$ positive constant such that $c_{0,k}=1$, we can use~\eqref{DeltaInComp} to compute
\begin{eqnarray*}
[\Delta_{\sigma},\partial_x^k]\phi
&=&
\Delta_{\sigma}\partial_x^k\phi -\partial_x^k\left[ \frac{1}{\sqrt{|\sigma|}}\partial_{i}\left(\sigma^{ij} \sqrt{|\sigma |} \partial_{j}\phi\right ) \right]\;,
\\
&=&
\Delta_{\sigma}\partial_x^k\phi-\sum_{k_1+k_2=k}c_{k_1,k_2}\partial_x^{k_1} \left( |\sigma|^{-1/2} \right)\partial_x^{k_2}\left(\sigma^{ij} \sqrt{|\sigma |} \partial_{j}\phi\right )
\\
&=&
\Delta_{\sigma}\partial_x^k\phi-c_{0k}|\sigma|^{-1/2} \partial_i\sum_{k_1+k_2=k} c_{k_1,k_2}\partial_x^{k_1}\left(\sigma^{ij} \sqrt{|\sigma |}\right ) \partial_x^{k_2}\partial_{j}\phi
\\
&& -\sum_{k_1+k_2=k \;,\; k_1\neq 0}c_{k_1,k_2}\partial_x^{k_1} \left( |\sigma|^{-1/2} \right)\partial_x^{k_2}\left(\sigma^{ij} \sqrt{|\sigma |} \partial_{j}\phi\right )
\\
&=&
\Delta_{\sigma}\partial_x^k\phi-|\sigma|^{-1/2} \partial_i \left(c_{0k}\sigma^{ij} \sqrt{|\sigma |} \partial_{j}\partial_x^{k}\phi\right )
\\
&&
-|\sigma|^{-1/2} \partial_i\sum_{k_1+k_2=k\;,\; k_1\neq 0} c_{k_1,k_2}\partial_x^{k_1}\left(\sigma^{ij} \sqrt{|\sigma |}\right ) \partial_x^{k_2}\partial_{j}\phi
\\
&& -\sum_{k_1+k_2=k \;,\; k_1\neq 0}c_{k_1,k_2}\partial_x^{k_1} \left( |\sigma|^{-1/2} \right)\partial_i\sum_{\tilde k_1+\tilde k_2=k_2} c_{\tilde k_1,\tilde k_2}\partial_x^{\tilde k_1}\left(\sigma^{ij} \sqrt{|\sigma |}\right ) \partial_x^{\tilde k_2}\partial_{j}\phi\;,
\end{eqnarray*}
since the two terms in the first line of the last equality cancel out, we finally arrive at
\begin{eqnarray}
\nonumber
-a^{2}[\square_{\metric},\partial_x^k] \phi^A &=& |\sigma|^{-1/2} \partial_i\sum_{k_1+k_2=k\;,\; k_1\neq 0} c_{k_1,k_2}\partial_x^{k_1}\left(\sigma^{ij} \sqrt{|\sigma |}\right ) \partial_x^{k_2}\partial_{j}\phi^A
\\
\label{commutator}
&&+\sum_{k_1+k_2=k \;,\; k_1\neq 0}c_{k_1,k_2}\partial_x^{k_1} \left( |\sigma|^{-1/2} \right)\partial_i\sum_{\tilde k_1+\tilde k_2=k_2} c_{\tilde k_1,\tilde k_2}\partial_x^{\tilde k_1}\left(\sigma^{ij} \sqrt{|\sigma |}\right ) \partial_x^{\tilde k_2}\partial_{j}\phi^A\;.
\end{eqnarray}

We can now use Jacobi's formula to write, given $l\in\bbR$, 
\begin{equation}
\partial_x |\sigma|^{l} = l |\sigma|^{l} \sigma^{ij}\partial_x\sigma_{ij}\;.
\end{equation}
 Recall the well known fact that
\begin{equation}
\partial_x \sigma^{ij} = -\sigma^{kj}\sigma^{is}\partial_x\sigma_{sk}\;.
\end{equation}
 By direct inspection of~\eqref{commutator} we see that: i)  all terms on the right hand side contain at least one factor involving derivatives of $|\sigma|^l$ or of $\sigma^{ij}$ which, according to the previous identities and~\eqref{smallSigmaK} can be bounded, in $L^{\infty}(\bbR^n)$, by $CC_{\sigma}$; ii) all terms contain derivatives $\partial_x^{k}\partial_{j}\phi^A$, with $0\leq k\leq K$, all of which can be bounded, in $L^{2}(\bbR^n)$, by the energy $\mathcal{E}_{K}$; 
  iii) Finally, there are also terms involving factors of $|\sigma|^l$ or $\sigma^{ij}$ which are clearly bounded, in $L^{\infty}(\bbR^n)$. We thus conclude that 
\begin{equation}
\lp{[\square_{\metric},\partial_x^k]\phi^A(t,\,\cdot\,)}{L^2(\bbR^n)}\leq C  a^{-2}(t) C_{\sigma}\mathcal{E}_{K}(t)\;,
\end{equation}
 and consequently 
 \begin{equation} 
 \label{commEst}
\int_{t_0}^{T}a(t)\lp{[\square_{\metric},\partial_x^k]\phi^A }{\Lx} dt \leq CC_{\sigma} \int_{t_0}^{T}\frac{\mathcal{E}_{K}(t)}{a(t)} dt\;.
 \end{equation}

We are now ready to close our bootstrap argument. From estimates ~\eqref{energyEst2},~\eqref{nonlinEst} and~\eqref{commEst}, it follows that  
  \begin{eqnarray*} 
\mathcal{E}_{K}(t)
 & \leq& 
\sup_{t_0\leq t<T}\mathcal{E}_{K}(t) 
\\
&\leq &
C\mathcal{E}_{K}(t_0)
 +d(K+1)4MCC_0 \int_{t_0}^{T}\frac{\mathcal{E}_{K}(t)}{a(t)} dt+d(K+1)CC_\sigma \int_{t_0}^{T}\frac{\mathcal{E}_{K}(t)}{a(t)} dt
 \\
&\leq& C_2 \left(\mathcal{E}_{K}(t_0)
 +C_0 \int_{t_0}^{T}\frac{\mathcal{E}_{K}(t)}{a(t)} dt+C_\sigma \int_{t_0}^{T}\frac{\mathcal{E}_{K}(t)}{a(t)} dt\right)  \;.
   \end{eqnarray*} 
Using Gr\"onwall's inequality leads to
     \begin{equation} 
 \label{energyEst4}
\mathcal{E}_{K}(t)  \leq  C_2 \mathcal{E}_{K}(t_0) \exp \left( C_2(C_0+C_\sigma) \int_{t_0}^{T}\frac{1}{a(t)} dt \right) \;.
  \end{equation}
By the integrability condition \eqref{intCond} and
choosing $C_0+C_\sigma$ sufficiently small, we can ensure 
\[\exp \left( C_2(C_0+C_\sigma) \int_{t_0}^{T}\frac{1}{a(t)} dt \right) < 2 .\]
Applying this, \eqref{initial_cond_boot}, and taking supremum in~\eqref{energyEst4} then yields 
 \begin{equation} 
 \label{energyEst5}
  \sup_{t_0\leq t<T}\mathcal{E}_{K}(t) \leq 2MC_0\;,
  \end{equation}
which corresponds to a strict improvement of the bootstrap assumption~\eqref{bootAss}, from which we can conclude that $T=+\infty$.

\section{Sharp decay estimates}
\label{secDecay}
We will establish sharp  decay upper bounds for the global solutions constructed in the previous section. 

Dropping the capital latin indices to simplify notation, we see that our wave equation can be written in the form 
\bea
\label{waveEqDecay}
\partial_t(a^n\partial_t \phi)=a^{n-2}\left({\Delta_{\sigma}}\phi- a^2\tilde\nN^{\alpha \beta}\partial_{\alpha} \phi \partial_{\beta} \phi\right)\;.
\eea
From~\eqref{smallSigmaK} and~\eqref{dxBound} we see that 
\bea
\nonumber
\lp{\Delta_{\sigma}\phi}{L^{\infty}}
&\leq& 
\lp{|\sigma|^{-1/2}\partial_i\left(\sigma^{ij}\sqrt{|\sigma|}\right)\partial_j\phi}{L^{\infty}} 
+\lp{\sigma^{ij}\partial_i\partial_j\phi}{L^{\infty}} 
\\
\label{deltaControl}
&\lesssim& C_{\sigma}C_0\;,
\eea
while~\eqref{NtermsFinal},~\eqref{energyEst5} and Sobolev embedding imply 
\bea
\lp{a^2\tilde\nN^{\alpha \beta}\partial_{\alpha} \phi \partial_{\beta} \phi}{L^{\infty}}\lesssim C_0^2\;.
\eea 
Then, integrating~\eqref{waveEqDecay}, we conclude 
\bea 
a^n(t)\partial_t\phi (t,x)  = a^n(t_0)\partial_t\phi (t_0,x) +\int_{t_0}^t a^{n-2}(s)\left({\Delta_{\sigma}}\phi- a^2\tilde\nN^{\alpha \beta}\partial_{\alpha} \phi \partial_{\beta} \phi\right) (s,x) ds \; ,
\eea
from which, in view of the previous estimates, it follows that 
\bea 
\label{dtEst}
|a^n(t)\partial_t\phi (t,x)|  \lesssim  C_0 + C_0 \int_{t_0}^t  a^{n-2}(s) ds \; ,
\eea
therefore
\bea 
\nonumber
|\partial_t\phi (t,x)|  &\lesssim& C_0 \left( 1 +  \int_{t_0}^t  a^{n-2}(s) ds \right) a^{-n}(t) 
\\
\label{dtEst2}
&\lesssim&  C_0 \left( \int_{t_0}^t  a^{n-2}(s) ds \right) a^{-n}(t) 
\;.
\eea

Note that if we consider the de Sitter case $a(t)=e^{Ht}$ we immediately obtain (for $n\geq 2$)
\bea 
|\partial_t\phi (t,x)|  \lesssim C_0 e^{-2Ht} \;, 
\eea
while in the case $a(t)=t^p$, $p>1$, we have 
\bea 
|\partial_t\phi (t,x)|  \lesssim C_0 t^{-2p+1} \;. 
\eea

If we commute the wave equation with spatial derivatives we obtain 
\bea
\label{waveEqDecayDx}
\partial_t(a^n\partial_t \partial^k_x \phi)=a^{n-2}\left(\partial^k_x{\Delta_{\sigma}}\phi- a^2\partial^k_x\tilde\nN^{\alpha \beta}\partial_{\alpha} \phi \partial_{\beta} \phi\right)\;.
\eea
Relying once again on~\eqref{dxBound} and ~\eqref{smallSigmaK}, a simple adaptation of~\eqref{deltaControl} shows that, for $k<K-\frac{n}{2}$, 
$$
\lp{\partial^k_x\Delta_{\sigma}\phi}{L^{\infty}}\lesssim C_{\sigma}C_0\;.
$$
while~\eqref{NtermsFinal} and Sobolev embedding give us, also  for $k<K-\frac{n}{2}$,
$$
\lp{a^2\partial^k_x\tilde\nN^{\alpha \beta}\partial_{\alpha} \phi \partial_{\beta} \phi}{L^{\infty}}\lesssim C_0\;. 
$$
So, by integrating~\eqref{waveEqDecayDx} we conclude that 
\bea 
\label{intaEst}
|\partial_t\partial_x^k\phi (t,x)|  \lesssim C_0 \left( \int_{t_0}^t  a^{n-2}(s) ds \right) a^{-n}(t) \;, 
\eea
provided that $k<K-\frac{n}{2}$. \\

Next we focus on constructing the limiting function $\phi_{\infty}$ in the case of a general expanding factor. To do that we start by noticing that since $a(t)$ is positive and increasing we have 
\bea
\label{intCond2}
\int_{t}^{\infty} \left( \int_{t_0}^s
a^{n-2}(u) du \right) a^{-n}(s)ds\leq \int_{t}^{\infty}a^{-1}(s)ds\int_{t_0}^{\infty}a^{-1}(u)du
\leq C \int_{t}^{\infty}a^{-1}(s)ds\;.
\label{non_sharp}
\eea
We can then set
\bea 
\phi_{\infty}(x) := \phi(t_0,x)+\lim_{t\rightarrow \infty} \int_{t_0} ^t  \partial_t\phi(s,x) ds \;,
\eea
which is well defined in view of~\eqref{dtEst2},~\eqref{intCond2} and the fact that $a^{-1}$ is integrable.

Then
%
\bea 
|\phi(t,x)-\phi_{\infty}(x)|\leq \int_{t}^{\infty} |\partial_t\phi(s,x)|  ds \leq C \int_{t}^{\infty}a^{-1}(s)ds \;, \label{sharp_bound_gen}
\eea
and 
\bea
	\lp{\phi(t,\,\cdot\,)-\phi_{\infty}}{L^{\infty}}\lesssim C_0 \int_t^{\infty}a^{-1}(s)ds\rightarrow 0 \;,
\eea
as $t\rightarrow \infty$.
In particular, this gives
\bea 
\lp{\phi(t_0,\,\cdot\,)-\phi_{\infty}}{L^{\infty}}\lesssim C_0\;.
\eea

Similarly, we can define  
\bea 
\phi_{x,\infty}(x) := \partial_x\phi(t_0,x)+\lim_{t\rightarrow \infty} \int_{t_0} ^t  \partial_t\partial_x\phi(s,x) ds \;,
\eea
and obtain
\bea 
|\partial_t\phi(t,x)-\phi_{x,\infty}(x)|\leq \int_{t_0}^t |\partial_t\partial_x\phi(s,x)|  ds \lesssim C_0 \int_t^{\infty}a^{-1}(s)ds\rightarrow 0 \; , 
\eea
as $t\rightarrow \infty$. From this uniform convergence and the already established convergence of $\phi(\,\cdot\,,x)$, as $t\rightarrow \infty$, we conclude that 
$$\phi_{x,\infty}=\partial_x\phi_{\infty}\;.$$

It is now easy to conclude by induction that, for all $k<K-\frac{n}{2}$, we have 
\bea 
\lp{\partial_x^k\left(\phi(t,\,\cdot\,)-\phi_{\infty}\right)}{L^{\infty}}\lesssim C_0 \int_t^{\infty}a^{-1}(s)ds\rightarrow 0 \;, 
\eea
as $t\rightarrow \infty$. \\

The quantitative decay estimates~\eqref{decay2function1} and~\eqref{decay2function1}, which are specific to the power law case $a(t)=t^p$, $p>1$, and the de Sitter case $a(t)=e^{Ht}$, $H>0$, respectively, now follow easily, by using the corresponding expansion factors in the previous procedure.

\section*{Acknowledgements}
This work was partially supported by FCT/Portugal through CAMGSD, IST-ID ,
projects UIDB/04459/2020 and UIDP/04459/2020,  by FCT/Portugal and CERN through project CERN/FIS-PAR/0023/2019 and through the FCT fellowship CEECIND/00936/2018 (A.F.).

\end{document}